\documentclass[twocolumn,pra,amsmath,nofootinbib]{revtex4} 
\usepackage[utf8]{inputenc}
\usepackage[pdftex,colorlinks]{hyperref}
\usepackage{amsthm}
\hypersetup{urlcolor=blue}
\newtheorem*{theorem}{Theorem}
\newcommand{\sgn}{\mathrm{sgn}}
\begin{document}
\title{Determination of hidden variable models reproducing the spin-singlet}
\author{Antonio \surname{Di Lorenzo}}
\affiliation{Instituto de F\'{\i}sica, Universidade Federal de Uberl\^{a}ndia, 38400-902 Uberl\^{a}ndia, Minas Gerais, Brazil}
\begin{abstract}
The experimental violation of Bell inequality establishes necessary but not sufficient conditions that any theory must obey. 
Namely, a theory compatible with the experimental observations can satisfy at most two of the three hypotheses at the basis of Bell\rq{}s theorem: 
free will, no-signaling, and outcome-Independence. 
Quantum mechanics satisfies the first two hypotheses but not  the latter. 
Experiments not only violate Bell inequality, but show an excellent agreement with quantum mechanics. 
This fact restricts further the class of admissible theories. 
In this work, the author determines the form of the hidden-variable models that  
reproduce the quantum mechanical predictions for a spin singlet while 
satisfying both the hypotheses of free will and no-signaling. 
Two classes of hidden-variable models are given as an example, and a general recipe to build 
infinitely many possible models is provided. 
\end{abstract}
\maketitle

\section{Introduction}
Are there theories more fundamental than quantum mechanics?  
Since the groundbreaking work of Bell \cite{Bell1964} 
this question has garnered increasing attention. 
These purported more fundamental theories are known as hidden-variable models, 
since they rely on the existence of parameters, the hidden variables, 
that are distinct from the wave function and from the classical observables  
(as energy, positions, etc.).  

A considerable result was achieved by Bell and perfected by others \cite{Bell1964,Clauser1969,Clauser1974} who showed that a whole family of such theories 
could be experimentally tested even though no explicit hypothesis about their mathematical structure nor about the additional parameters was made,  
by requiring, instead, that the models lead to probabilistic predictions satisfying some \lq\lq{}reasonable\rq\rq{} assumptions.  
These assumptions, discussed at length below, are known as Measurement-Independence, Setting-Independence, and Outcome-Independence. 
The first one may be justified invoking both \cite{DiLorenzo2012b} the impossibility of action-at-a-distance and the independence of the measurement 
settings from any variables (\lq\lq{}free will\rq\rq{}); the second one is a consequence of the impossibility of superluminal signaling; 
the third assumption, however, is more difficult to justify \cite{Jarrett1984,Shimony1990}. 
More recently, Leggett \cite{Leggett2003} demonstrated the incompatibility of quantum mechanics with all models 
satisfying Measurement-Independence and a stronger form of Setting-Independence, the compliance with Malus\rq{}s law. 

Experiments \cite{Rowe2001,Lee2011} not only show a clear violation of both Bell and Leggett inequalities, but they reproduce accurately the 
predictions of quantum mechanics, since the discrepancies can be explained by the unavoidable imperfections of preparation and measurement. 
Thus, the constraints put on hidden variable models by the current experimental evidence are even stricter than the simple incompatibility with 
at least one of the three hypotheses of Bell (or of the two hypotheses of Leggett): after averaging over the hidden variables the predictions 
of quantum mechanics must be reproduced, in order for the theory to be admissible.

In the present paper, we consider models satisfying both Measurement-Independence and Setting-Independence, 
so that the principles of \lq\lq{}free will\rq\rq{} and no-signaling are satisfied. By building upon a recent theorem \cite{Colbeck2008,Branciard2008}, 
we shape the form of all such hidden-variable theories that are compatible with quantum mechanics and hence with experiments.

\section{The system and the goal}
The system of interest is a pair of particles in a spin-singlet configuration which fly to space-separated locations. We use the language of spin, rather than polarization of light, since the 
formulas are slightly more compact. The events consist in the determination of the spin projection along a given axis for each particle. We choose units such that the outcomes for each particle are 
$\sigma,\tau\in\{-1,1\}$.  
The measured observables are the spin projections $\mathbf{a}\cdot \hat{\mathbf{S}}_1$ and $\mathbf{b}\cdot \hat{\mathbf{S}}_2$, 
which we shall 
indicate simply by $\mathbf{a}$ and $\mathbf{b}$. For brevity, we write the conditional probability of observing the 
outcome $\{\sigma,\tau\}$ for given values of $\mathbf{a},\mathbf{b}$ and 
hidden variables $\lambda$ as $P(\sigma,\tau|\lambda,\mathbf{a},\mathbf{b})\equiv P_{\sigma,\tau}(\lambda,\mathbf{a},\mathbf{b})$. Quantum mechanics predicts that 
\begin{equation}
P^{QM}_{\sigma,\tau}(\psi,\mathbf{a},\mathbf{b})=\frac{1}{4}\left[1-\sigma\tau \mathbf{a}\cdot\mathbf{b}\right],
\end{equation}
where $\psi$ describe the preparation of two particles in a singlet state. 
Our goal is to determine a positive measure $d\mu$ and a conditional probability such that integration over the hidden variables yields $P^{QM}$, namely 
 \begin{equation}\label{eq:goal}
\int d\mu(\lambda|\mathbf{a},\mathbf{b}) P_{\sigma,\tau}(\lambda,\mathbf{a},\mathbf{b})= 
\frac{1}{4}\left[1-\sigma\tau \mathbf{a}\cdot\mathbf{b}\right].
\end{equation}
The preparation $\psi$ is henceforth omitted, and it is understood that it appears as a prior in all the probabilities. 

\section{Hypotheses at the basis of Bell and Leggett inequalities}
The models excluded by Bell inequality rely on three hypotheses: 
Measurement-Independence  (which we refer to as Uncorrelated Choice), 
Outcome-Independence (for which we propose the more descriptive term ``Reducibility of Correlations"), 
and Setting-Independence.
The models excluded by Leggett inequality rely on Uncorrelated Choice and 
compliance with Malus's law. 
As one or more assumptions must be violated, 
we shall briefly discuss the physical meaning of these assumptions, in order to 
individuate the least problematic hypotheses to drop. 
In the following, we shall use the word ``locality\rq\rq{}, by which we mean simply the impossibility of superluminal signaling. 

Uncorrelated Choice (UC), sometimes called Measurement-Independence, 
means that the distribution of the $\lambda$ and the 
settings of the detectors are uncorrelated. 
If one thinks of $\lambda$ as a set of parameters attached to 
the physical system, then Uncorrelated Choice follows from locality.
 However, it may happen that $\lambda$ 
is correlated with the choice of the observables to be detected, due to some past common 
cause \cite{Brans1988}, and in this case Uncorrelated Choice 
can be violated even though locality holds \cite{Brans1988,Hall2010,DiLorenzo2012b,DiLorenzo2012d}. 
Indeed, if one considers that the choice of settings, be it done by an automatic random mechanism 
or by a conscious being, can be influenced by events in the past light-cone of either station $A$ or $B$, and considers further that 
these light-cones have an intersection between themselves and with the past light-cone of the entangling apparatus, 
it is  possible, 
in principle, that there are correlations between the hidden variables and the choice of settings. 
Given our actual knowledge, however, this is a remote possibility. 
Usually, it implies a limitation of free will, or a conspiracy 
of sorts, but there is a possibility that what appears a conspiracy today 
is but a manifestation of some fundamental law.

Reducibility of Correlations (RC), 
known also as Outcome-Independence, means that the conditional probability of 
the outcome $\tau$, given $\lambda$ and given that the outcome of the measurement 
of $\mathbf{a}$ is $\sigma$, does not depend on the latter, namely 
$P_\tau(\lambda,\mathbf{a},\mathbf{b},\sigma)=P_\tau(\lambda,\mathbf{a},\mathbf{b})$, 
so that the joint probability is 
$P_{\sigma,\tau}(\lambda,\mathbf{a},\mathbf{b},\sigma)=
P_\sigma(\lambda,\mathbf{a},\mathbf{b})P_\tau(\lambda,\mathbf{a},\mathbf{b})$.
Hence, if the parameters $\lambda$ could be accessed, by either measuring them or fixing them, there would be no correlations. 
After averaging over $\lambda$, however, correlations appear. Thus Reducibility of Correlations means 
that the quantum correlations emerge from the ignorance of some more fundamental parameters.
In order to check whether Reducibility of Correlations holds, the observer at $A$ must 
calculate the conditional probability $P_\sigma(\tau,\mathbf{a},\mathbf{b},\lambda)$ (still assuming that $\lambda$ can be accessed by $A$), and check whether it varies 
when $\tau$ varies, while the other parameters are fixed. In order to do so, $A$ must have access to the remote information $\mathbf{b},\tau$, which 
$B$ can send only at  a speed not exceeding the speed of light.
Hence, violating Reducibility of Correlations 
does not imply action-at-a-distance, nor the possibility of instantaneous communication, 

Setting-Independence (SI) means that the marginal probability of observing the event 
$\sigma$ at $A$, 
for a given $\lambda$, does not depend on the setting $\mathbf{b}$, 
namely $P_\sigma(\lambda,\mathbf{a},\mathbf{b})=P_\sigma(\lambda,\mathbf{a})$. 
It may seem that the violation 
of Setting-Independence gives the possibility of instantaneous signaling, 
so that locality implies Setting-Independence.\footnote{Some authors, indeed, identify 
Setting-Independence with no-signaling, the impossibility of instantaneous communication, while they reserve the term ``locality" sometimes to mean Reducibility of Correlations, other times to mean both Reducibility of Correlations and 
Setting-Independence, and 
other times still to refer to the three hypotheses 
Uncorrelated Choice, Reducibility of Correlations, and Setting-Independence.}
However, this is true only if $\lambda$ 
has a fixed known value, or if it can be completely determined by a measurement 
at location $A$ (or $B$). 

Finally, compliance with Malus's law requires that the hidden-variables consist in a unit-vector 
such that the marginal probability is 
$P_\sigma(\mathbf{u},\mathbf{a},\mathbf{b})=\left(1+\sigma\mathbf{a}\cdot\mathbf{u}\right)/2$. 
Therefore, this hypothesis is a special case of Setting-Independence. We remark that this hypothesis tries to give 
a physical meaning to the hidden variables, assuming that they are made of unit vectors in such a way that each spin (or photon) 
possesses a well defined polarization, in such a way that, if the polarization could be fixed, the ordinary Malus\rq{}s law would be obeyed. 

By relaxing the hypothesis of Uncorrelated Choice, e.g., 
it is possible to violate both Bell and Leggett inequalities \cite
{Brans1988,Hall2010,DiLorenzo2012b,DiLorenzo2012d},   
a necessary condition to reproduce on average the results of quantum mechanics. 
Other possibilities explored in the literature consist in violating both Uncorrelated Choice 
and Reducibility of Correlations \cite{Cerf2005,Groblacher2007b,Brunner2008},  
or only Setting-Independence\cite{Toner2003,Pawlowski2010}.

\section{Examples}
Now, let us construct a family of  models compatible with quantum mechanics. 
We consider only models obeying the hypotheses of Uncorrelated Choice and Setting-Independence, 
since the violation of either hypothesis may have controversial implications.
A big help is provided by the trivial-marginals theorem, 
derived (under assumptions slightly stronger than the strictly necessary ones) 
by Colbeck and Renner \cite{Colbeck2008} and Branciard \emph{et al.} \cite{Branciard2008}, 
and rederived (under minimal assumptions) in Appendix \ref{app:trivmarg}. 
This theorem states that all hidden variable models that satisfy Uncorrelated Choice and 
Setting-Independence while reproducing the quantum mechanical predictions, must have a 
$\lambda$-conditioned probability of the form 
\begin{align}\label{eq:prob00}
&P_{\sigma,\tau}(\lambda,\mathbf{a},\mathbf{b}) =\ \frac{1}{4}\biggl\{
1-\sigma\tau \bigl[\mathbf{a}\cdot\mathbf{b}-C(\lambda,\mathbf{a},\mathbf{b})\bigr]\biggr\},
\end{align}
where $C$ has a zero average with the weight $d\mu(\lambda)$ (which represents the probability distribution 
of the $\lambda$), and $C(\lambda,\mathbf{a},\mathbf{a})=0$ 
with the exclusion of the subsets of $\lambda$ where $\mu(\lambda)$ is identically zero. 
In particular all models satisfying Malus's law are excluded by the theorem. 
In other words, assuming Uncorrelated Choice and Malus's law  
(which is a special case of Setting-Independence) 
results in theories incompatible with quantum mechanics. 
The function $C(\lambda,\mathbf{a},\mathbf{b})$ represents the excess or defect correlations 
(with respect to the quantum mechanical correlations) attributable to the hidden variables. 
Indeed, if $\lambda$ could be fixed, either by the specification of a suitable preparation procedure or by 
post-selection provided a prescription for its measurement is given, then the observed spin-spin correlations for 
detectors oriented along $\mathbf{a},\mathbf{b}$ would be $\mathrm{Corr}(\lambda,\mathbf{a},\mathbf{b})=
-\mathbf{a}\cdot\mathbf{b}+C(\lambda,\mathbf{a},\mathbf{b})$. As $C$ must vanish on average, it 
can take both positive and negative values for different $\lambda$s, thus the correlations, at the hidden variable level, 
may be stronger than the quantum mechanical correlations. 

No actual examples of models satisfying the trivial-marginals theorem were made so far, 
and we proceed to fill this gap. 
First, we notice that not any choice of $C$ leads to a positive-defined probability. 
For instance, choosing 
\begin{equation}\label{eq:wrongtrial}
C=\sqrt{1-(\mathbf{a}\cdot\mathbf{b})^2}\ G(\lambda)
\end{equation} 
leads to regions of negative probabilities for any function $G$ having zero average. 
An important result discussed in the next section will be to provide a recipe for 
building up all admissible functions $C$. 
As an example, we choose, e.g.,  
\begin{equation}\label{eq:first}
C(\lambda,\mathbf{a},\mathbf{b})=\left[1-(\mathbf{a}\cdot\mathbf{b})^2\right]G(\lambda), 
\end{equation} 
with $|G(\lambda)|<1/2$ and $\int d\lambda \mu(\lambda) G(\lambda)=0$. 
It is easy to check that the probability in Eq.~\eqref{eq:prob00} is always positive and that upon averaging 
over $\lambda$ Eq.~\eqref{eq:goal} is satisfied. 
Thus we have constructed a family of hidden-variable theories that reproduce quantum mechanics. 
Notice that neither the hypothesis of Reducibility of Correlations, needed in order to derive Bell inequality, 
nor the hypothesis of compliance with Malus's law, needed for Leggett inequality, are satisfied. 

Another family of local models, i.e. requiring and allowing no instantaneous communication between the two wings, is obtained by choosing 
\begin{equation}\label{eq:second}
C(\lambda,\mathbf{u},\mathbf{a},\mathbf{b})=-\mathbf{a}\cdot\mathbf{b} \left[ (\mathbf{a}\cdot\mathbf{u})^2-(\mathbf{b}\cdot\mathbf{u})^2\right]^2 G(\lambda), 
\end{equation}
with, as before, $|G(\lambda)|\le 1/2$ and having zero average, while $\mathbf{u}$ is 
a unit-vector hidden variable. 

It
 can be shown that the model of Cerf \emph{et al.} \cite{Cerf2005} can be reduced 
to the form 
\begin{align}
\label{eq:Cerfmu2}
\mu(\mathbf{u},\mathbf{v}|\mathbf{a},\mathbf{b})=&
\frac{1}{(4\pi)^2} ,\\ 
\nonumber
P_{\sigma,\tau}(\mathbf{u},\mathbf{v},\mathbf{a},\mathbf{b})=&
\frac{1}{4}\biggl[1 -\sigma\tau\, \sgn{(\mathbf{u}\!\cdot\!\mathbf{a})}\, \sgn{(\mathbf{n}_+\!\cdot\!\mathbf{b})}
\\
\label{eq:Cerfp2}
&\times \frac{1+x_{\mathbf{a},\mathbf{u},\mathbf{v}}+y_{\mathbf{b},\mathbf{u},\mathbf{v}}-x_{\mathbf{a},\mathbf{u},\mathbf{v}}y_{\mathbf{b},\mathbf{u},\mathbf{v}}}{2}\biggr],
\end{align}
where 
$x_{\mathbf{a},\mathbf{u},\mathbf{v}}=\sgn{(\mathbf{u}\cdot\mathbf{a})}\ \sgn{(\mathbf{v}\cdot\mathbf{a})}$, 
$y_{\mathbf{b},\mathbf{u},\mathbf{v}}=\sgn{(\mathbf{n}_+\cdot\mathbf{b})}\ \sgn{(\mathbf{n}_-\cdot\mathbf{b})}$, 
and $\mathbf{n}_\pm=\mathbf{u}\pm\mathbf{v}$. 
The reader can verify that this model reproduces the quantum mechanical predictions 
while it satisfies UC and SI, but violates RC and thus falls within the family of models we are interested in. 
\section{Main theorem}
While the examples above were found by trial and error, a careful analysis of the 
presence or lack of negative regions for the probabilities leads to the main result of the 
present paper.
\begin{theorem}
The function $C$ in Eq.~\eqref{eq:prob00} is of the form 
\begin{equation}
C(\lambda,\mathbf{a},\mathbf{b})=
\left[1+\mathbf{a}\cdot\mathbf{b}\right]^{s_+}\left[1-\mathbf{a}\cdot\mathbf{b}\right]^{s_-} 
G(\lambda,\mathbf{a},\mathbf{b}), 
\end{equation}
with 
\begin{subequations}
\begin{align}
\label{constr1} 
&0<|G(\lambda,\mathbf{a},\pm\mathbf{a})|<\infty, \ \text{for } \lambda\in D^\pm_{\mathbf{a}}, \mu(D^\pm_{\mathbf{a}})>0,\\
&
\label{constr2} \int d\mu(\lambda) G(\lambda,\mathbf{a},\mathbf{b}) =0, \\ 
&
\nonumber
\frac{-1}{\left[1-\mathbf{a}\cdot\mathbf{b}\right]^{s_--1}
\left[1+\mathbf{a}\cdot\mathbf{b}\right]^{s_+}}\le G(\lambda,\mathbf{a},\mathbf{b})\\
\label{constr3} 
&\hspace{2.7cm}\le \frac{1}{\left[1-\mathbf{a}\cdot\mathbf{b}\right]^{s_-}\left[1+\mathbf{a}\cdot\mathbf{b}\right]^{s_+-1}}, \\
&
\label{constr4} s_+\ge 1 \ , \  s_-\ge 1, \\
&
\label{constr5} |G(\lambda,\mathbf{a},\pm\mathbf{a})|\le 1/2^{s_\pm} \ \mbox{if}\ s_{\mp}=1.
\end{align}
\end{subequations}
\end{theorem}
\begin{proof}
In order to satisfy $C(\lambda,\mathbf{a},\pm\mathbf{a})=0$, 
the function $C$ must be of the form
\begin{equation}
C(\lambda,\mathbf{a},\mathbf{b})=
\left[1+\mathbf{a}\cdot\mathbf{b}\right]^{s_+}\left[1-\mathbf{a}\cdot\mathbf{b}\right]^{s_-} 
G(\lambda,\mathbf{a},\mathbf{b}), 
\end{equation}
with $s_+>0,s_->0$ and $0<|G(\lambda,\mathbf{a},\pm\mathbf{a})|<\infty$. 
This is a Frobenius-like expansion, with $s_\pm$ 
determining how fast the function vanishes for $\mathbf{a}\cdot\mathbf{b}=\mp 1$, 
so that Eq.~\eqref{constr1} follows by definition, 
with $D^\pm_\mathbf{a}$ a domain of $\lambda$ 
having non-zero measure (if $G$ is identically zero almost everywhere 
when $\mathbf{a}\cdot\mathbf{b}=\pm 1$ we can then redefine $s_\pm$ and $G$). 
Equation \eqref{constr2} follows from $\int d\mu(\lambda) C(\lambda,\mathbf{a},\mathbf{b}) =0$. 
The positivity of the probability implies the inequalities in Eq.~\eqref{constr3}. These inequalities also guarantee that none of the four probabilities exceeds one, since 
one can readily verify the more precise inequality 
$P_{\sigma,\tau}(\lambda,\mathbf{a},\mathbf{b})\le 1/2, \forall \sigma,\tau,\lambda,\mathbf{a},\mathbf{b}$. 
Furthermore, if we let $\mathbf{a}\cdot\mathbf{b} = \pm(1-\varepsilon)$, we have that 
the probability of $\{\sigma,\pm\sigma\}$ is, to lowest order 
\begin{equation}\label{eq:exp}
P_{\sigma,\pm\sigma}\simeq\frac{\varepsilon}{4}\left[
1\mp 2^{s_\pm}\varepsilon^{s_{\mp}-1} G(\lambda,\mathbf{a},\pm\mathbf{a})\right] , 
\end{equation}
therefore, remembering that $G(\lambda,\mathbf{a},\pm\mathbf{a})$ changes sign when varying $\lambda$, 
in order for the probability to be positive we must have $s_{\pm}\ge 1$, proving Eq.~\eqref{constr4}. Finally, assuming that 
$s_\pm=1$, Eq.~\eqref{eq:exp} implies that 
$|G(\lambda,\mathbf{a},\pm\mathbf{a})|\le 1/2^{s_\mp}$ for $s_{\pm}=1$, so that Eq.~\eqref{constr5} is proved.  
\end{proof}
The two families described by Eqs.~\eqref{eq:first} and \eqref{eq:second} have $s_{+}=s_{-}=1$. There may be pathological cases which are not captured 
by our theorem, e.g., it may happen that $G(\lambda,\mathbf{a},\mathbf{b})$ averages to zero for all $\mathbf{b}\neq \pm\mathbf{a}$, but that 
$G(\lambda,\mathbf{a},\pm\mathbf{a})$ has a constant sign for varying $\lambda$, so that 
Eqs.~\eqref{constr4} and \eqref{constr5} may not be satisfied. 
This behavior requires some essential non-analyticity, and we believe it is not physically interesting. 

Let us provide a constructive recipe to build families of hidden-variable models. 
We choose $s_+=s_-=s$ just for the sake of symmetry. 
Now we pick an arbitrary limited function $f(\lambda,\mathbf{a},\mathbf{b})$ having a finite value for 
$\int d\mu(\lambda) f(\lambda,\mathbf{a},\mathbf{b})$, where $\lambda$ may include vectors, scalars, discrete variables (in which case the integral is a sum).  
 We build the zero-average function 
$g(\lambda,\mathbf{a},\mathbf{b})=f(\lambda,\mathbf{a},\mathbf{b})-\int d\mu(\lambda\rq{}) f(\lambda\rq{},\mathbf{a},\mathbf{b})$. We consider the  supremum and infimum of $g$, $M$ and $m$. By construction $M>0$ and $m<0$. 
If they satisfy 
\begin{equation}
\label{eq:ineq0}
-\frac{(s-1/2)^{2s-1}}{s^s (s-1)^{s-1}} \le m<M\le 
\frac{(s-1/2)^{2s-1}}{s^s (s-1)^{s-1}},
\end{equation}
 then our job is done, since 
 \begin{align}
\nonumber
&\frac{-1}{(1-\mathbf{a}\cdot\mathbf{b})^{s-1}(1+\mathbf{a}\cdot\mathbf{b})^{s}}\le 
-\frac{(s-1/2)^{2s-1}}{s^s (s-1)^{s-1}}\\
&< 
\frac{(s-1/2)^{2s-1}}{s^s (s-1)^{s-1}}\le 
\frac{1}{(1-\mathbf{a}\cdot\mathbf{b})^{s}(1+\mathbf{a}\cdot\mathbf{b})^{s-1}}.
 \end{align}
Otherwise, we multiply $g$ by an appropriate factor, so that Eq.~\eqref{eq:ineq0} is satisfied. 
The resulting function $G(\lambda,\mathbf{a},\mathbf{b})
=[1-(\mathbf{a}\cdot\mathbf{b})^2]^s g(\lambda,\mathbf{a},\mathbf{b})$ satisfies all the hypotheses of the main theorem 
by construction, and we have built a hidden variable model. 

\section{Discussion}
There appears to be a contrast between the results presented in the paragraph above and those reported in two recent papers \cite{Pawlowski2010,Colbeck2011}. 
We shall briefly discuss these contrasts. 

Reference \cite{Pawlowski2010} claims that ``the assumed experimenter’s freedom to choose the
settings ensures that the setting information must be non-locally transferred even
when the SI condition is obeyed" and concludes that the work 
\lq\lq{}provides the general conditions that every non-local hidden 
variable theory has to satisfy in order to allow for violation of the CHSH inequality\rq\rq{}. 
These conclusions are evidently wrong, as we have provided models obeying Setting Independence and 
not only violating the CHSH inequality, but reproducing the full quantum mechanical predictions. 
As shown in Appendix \ref{app:pawlth}, the conclusions of Ref.~\cite{Pawlowski2010} are valid provided that 
they are restricted to models satisfying certain hypotheses. 
One of these hypotheses is that the conditional probability is not extracted from experimental data, but 
is simulated at location $A$ according to some algorithm. 
Here, instead, we are considering the possibility that, in addition to the wave-function, there exist 
further parameters $\lambda$ giving a finer description of the system. 
We agree with Ref. \cite{Pawlowski2010} that, if the (allegedly) experimentally accessible conditional probabilities 
were to be reproduced through an algorithm, then both $\mathbf{b}$ and $\tau$ should be transmitted to $A$.

On the other hand, the results presented here show that quantum mechanics can be extended 
through the specification of 
additional parameters $\lambda$, and that this extension has improved predictive power, since the function 
$C(\lambda,\mathbf{a},\mathbf{b})$ is non-zero, and consequently the predicted joint probability for 
given $\lambda$ differs from the quantum mechanical one: 
\begin{align}
\nonumber
P_{\sigma,\tau}(\lambda,\mathbf{a},\mathbf{b})=&
\frac{1}{4}\left\{1-\sigma\tau\left[\mathbf{a}\cdot\mathbf{b}-C(\lambda,\mathbf{a},\mathbf{b})\right]\right\}\\
\neq&  
P^{QM}_{\sigma,\tau}(\mathbf{a},\mathbf{b}).
\end{align}
This seems to contradict the findings of Ref.~\cite{Colbeck2011}. However, in Ref.~\cite{Colbeck2011},  
the impossibility to have an improved  predicted power refers 
to the marginal probability $P_\sigma(\lambda,\mathbf{a})$, not to the joint one 
$P_{\sigma,\tau}(\lambda,\mathbf{a},\mathbf{b})$. The models discussed in the present work 
predict marginal probabilities of $P_\sigma(\lambda,\mathbf{a})=1/2$ and hence do not contradict 
Ref.~\cite{Colbeck2011}. In other words, the apparent tension is due to the definition of 
`extension of quantum theory'.

\section{Conclusions}
In conclusion, we have established the form of all the hidden variable models able to reproduce the quantum mechanics of a spin-singlet by satisfying both 
the assumptions of \lq\lq{}free will\rq\rq{} and no-signaling, which correspond, respectively, to 
Uncorrelated Choice (Measurement-Independence) and Setting-Independence. 
By contrast, we have assumed the violation of Reducibility of Correlations, as this can never result in superluminal signaling, since it consists in   
the dependence of a conditional probability on a remote outcome, and as such it requires the communication 
of said outcome through means that are necessarily subluminal. 
Rather, the violation of Reducibility of Correlations implies that the quantum correlations cannot be 
attributed to the ignorance of the hidden parameters.

\section*{Acknowledgments}
This work was supported by Funda\c{c}\~{a}o de Amparo \`{a} Pesquisa do 
Estado de Minas Gerais through Process No. APQ-02804-10.

\appendix 
\section{The trivial-marginals theorem}\label{app:trivmarg}
We prove a theorem established in Ref.~\cite{Colbeck2008} assuming that the hidden variables 
can be written $\lambda=\lambda_L\bigcup\lambda_0\bigcup\lambda_R$, with $\lambda_{L,R}$ 
local parameters associated to the measurement at location $L,R$ admitting 
a factorable measure, 
and in Ref.~\cite{Branciard2008} for discrete variables only. 
The proof below relies on none of these additional assumptions. 
\begin{theorem}
All hidden-variable theories that satisfy Uncorrelated Choice and Setting-Independence, 
and that reproduce the quantum mechanical predictions 
for spin singlets predict conditional probabilities of the form 
\begin{equation}\label{eq:prob0}
P_{\sigma,\tau}(\lambda,\mathbf{a},\mathbf{b}) =\ \frac{1}{4}\biggl\{
1-\sigma\tau \bigl[\mathbf{a}\cdot\mathbf{b}-C(\lambda,\mathbf{a},\mathbf{b})\bigr]\biggr\},
\end{equation}
where
\begin{subequations}
\begin{align}
\label{eq:zeroav}
&\int d\mu(\lambda)C(\lambda,\mathbf{a},\mathbf{b})=\ 0,\\
\label{eq:correlations}
&C(\lambda,\mathbf{a},\pm\mathbf{a})=\ 0,\\
\label{eq:ineq}
&\left|\mathbf{a}\cdot\mathbf{b}-C(\lambda,\mathbf{a},\mathbf{b})\right|\le\ 1,
\end{align}
\end{subequations}
with $\mu(\lambda)$ a measure. 
\end{theorem}
\begin{proof}
Consider a hidden variable theory that tries to reproduce the quantum mechanical predictions 
for a spin singlet. It must satisfy Eq.~\eqref{eq:goal},  
with $d\mu(\lambda|\mathbf{a},\mathbf{b})$ a positive measure. Generally $d\mu(\lambda|\mathbf{a},\mathbf{b})=\mu(\lambda|\mathbf{a},\mathbf{b}) d\lambda$, 
and the positive normalized generalized function $\mu$ can be interpreted as the probability density of $\lambda$ for given $\mathbf{a},\mathbf{b}$.  
Measurement-Independence implies that the measure does not depend 
on the settings of the detectors, i.e., $d\mu(\lambda|\mathbf{a},\mathbf{b})=d\mu(\lambda)$ or 
$\mu(\lambda|\mathbf{a},\mathbf{b})=\mu(\lambda)$, 
Without loss of generality, we put 
\begin{equation}
P_{\sigma,\tau}(\lambda,\mathbf{a},\mathbf{b})=\frac{1}{4}\left[1-\sigma\tau \mathbf{a}\cdot\mathbf{b}+ 
\Delta_{\sigma,\tau}(\lambda,\mathbf{a},\mathbf{b})\right].
\end{equation}
The function $\Delta_{\sigma,\tau}(\lambda,\mathbf{a},\mathbf{b})$, by definition, satisfies 
\begin{align}\label{eq:avtozero}
\int d\mu(\lambda) \Delta_{\sigma,\tau}(\lambda,\mathbf{a},\mathbf{b}) =& 0,\\
\sum_{\sigma,\tau} \Delta_{\sigma,\tau}(\lambda,\mathbf{a},\mathbf{b})=& 0.
\end{align}
and it can be written as 
\begin{equation}
\Delta_{\sigma,\tau}(\lambda,\mathbf{a},\mathbf{b}) = 
\sigma A(\lambda,\mathbf{a},\mathbf{b})+
\tau B(\lambda,\mathbf{a},\mathbf{b})+\sigma\tau C(\lambda,\mathbf{a},\mathbf{b}),
\end{equation}
with all three functions satisfying Eq.~\eqref{eq:avtozero}. 
In particular, Eq.~\eqref{eq:zeroav} is satisfied. 
Setting-Independence requires that the marginal probability of observing the outcome 
$\sigma$ at detector $\mathbf{a}$ is not influenced by the direction $\mathbf{b}$ chosen for the other 
detector, and vice versa, namely
\begin{align}
P_{\sigma}(\lambda,\mathbf{a},\mathbf{b})\equiv\sum_\tau P_{\sigma,\tau}(\lambda,\mathbf{a},\mathbf{b})=& 
P_{\sigma}(\lambda,\mathbf{a}),\\ 
P_{\tau}(\lambda,\mathbf{a},\mathbf{b})\equiv\sum_\sigma P_{\sigma,\tau}(\lambda,\mathbf{a},\mathbf{b})=& 
P_{\tau}(\lambda,\mathbf{b}). 
\end{align}
Thus, we have that 
\begin{equation}
A(\lambda,\mathbf{a},\mathbf{b}) = A(\lambda,\mathbf{a}) \ , \ B(\lambda,\mathbf{a},\mathbf{b}) = 
B(\lambda,\mathbf{b}).
\end{equation}
In particular, quantum mechanics predicts perfect (anti)correlations when $\mathbf{a}=-\mathbf{b}$ 
($\mathbf{a}=\mathbf{b}$). This implies that (here it is fundamental that the measure is independent of 
$\mathbf{a}$ and $\mathbf{b}$)
\begin{align}
\label{eq:symm}
&A(\lambda,\mathbf{a})+B(\lambda,\mathbf{a})=A(\lambda,\mathbf{a})-B(\lambda,-\mathbf{a})=0,\\
\label{eq:corr}
&C(\lambda,\mathbf{a},\mathbf{a})=C(\lambda,\mathbf{a},-\mathbf{a})=0,
\end{align}
identically almost everywhere\footnote{`Almost everywhere' means in all subsets having non-zero measure. If 
$\lambda$ has a discrete distribution, so that $\mu(\lambda)$ is a sum of $\delta$-functions, `almost everywhere' 
means, paradoxically, only at the discrete values of $\lambda$.} in $\lambda$ and in $\mathbf{a}$. 
Equation~\eqref{eq:symm} is satisfied by $B(\lambda,\mathbf{a})=-A(\lambda,\mathbf{a})$, with 
$A(\lambda,-\mathbf{a})=-A(\lambda,\mathbf{a})$ an odd function of its second argument. 
Consider now values close to the perfect anticorrelation point, $\mathbf{b}=(\mathbf{a}+\boldsymbol{\delta})/\sqrt{1+\delta^2}$, 
with $\mathbf{a}\cdot\boldsymbol{\delta}=0$ and $|\boldsymbol{\delta}|\ll 1$.  
To first order, the probability $P_{\sigma,\sigma}(\lambda,\mathbf{a},\mathbf{b})$ reads 
\begin{equation}\label{eq:critical}
P_{\sigma,\sigma}(\lambda,\mathbf{a},\mathbf{a}+\boldsymbol{\delta}) = 
-\frac{1}{4}\boldsymbol{\delta}\cdot\left[
\sigma \frac{\partial A(\lambda,\mathbf{n})}{\partial \mathbf{n}} 
-\frac{\partial C(\lambda,\mathbf{a},\mathbf{n})}{\partial \mathbf{n}}\right]_{\mathbf{n}=\mathbf{a}} ,
\end{equation}
where $\mathbf{n}$ is a generic placeholder for a unit vector. 
Clearly, Eq.~\eqref{eq:critical} cannot be positive for all $\boldsymbol{\delta}$. If it is positive for a value 
$\boldsymbol{\delta}_0$, it will be negative for $\boldsymbol{\delta}=-\boldsymbol{\delta_0}$. 
The only possibility is that the term inside the brackets in \eqref{eq:critical} vanishes identically in 
$\lambda$ and $\mathbf{a}$ or that it is proportional to $\mathbf{a}$. 
By changing the sign of $\sigma$, summing and subtracting, we notice that the following two identities should hold
\begin{align}
\label{eq:id1}
&\left.\frac{\partial A(\lambda,\mathbf{n})}{\partial \mathbf{n}}\right|_{\mathbf{n}=\mathbf{a}}\!\!\!\!=\ 
f(\lambda,\mathbf{a})\mathbf{a},\\
\label{eq:id2}
&\left.\frac{\partial C(\lambda,\mathbf{a},\mathbf{n})}{\partial \mathbf{n}}\right|_{\mathbf{n}=\mathbf{a}}\!\!\!\!=\ g(\lambda,\mathbf{a})\mathbf{a}.
\end{align}
Since $A$ is an odd function and $\mathbf{a}$ a unit vector, Eq.~\eqref{eq:id1} implies that 
$A(\lambda,\mathbf{n})=0$: Indeed invariance requires that the dependence on the argument 
can be only of the form 
$A(\lambda,\mathbf{n})=A(\lambda,\mathbf{n}\cdot\mathbf{p}_j)$, where $\mathbf{p}_j$ are vectors 
either fixed or depending on the hidden variables (possibly being some of the hidden variables). 
We have then that 
\begin{equation}
A\left(\lambda,\frac{\mathbf{a}+\boldsymbol{\delta}}{\sqrt{1+\delta^2}}\right)-A(\lambda,\mathbf{a})\simeq 
\sum_j \mathbf{p}_j\cdot\boldsymbol{\delta} \left.\frac{\partial A(\lambda,x_j)}{\partial x_j}\right|_{x_j=\mathbf{a}\cdot\mathbf{p}_j}.
\end{equation}
This implies that 
\begin{equation}
\left.\frac{\partial A(\lambda,x_j)}{\partial x_j}\right|_{x_j=\mathbf{a}\cdot\mathbf{p}_j}=0 , 	
\end{equation}
and hence $A(\lambda,x_j)=f(\lambda)$. Since $A(\lambda,-x_j)=-A(\lambda,-x_j)$, $A(\lambda,x_j)=0$, i.e. the validity of Eq.~\eqref{eq:prob0} was proved. 
Then Eqs.~\eqref{eq:correlations} and \eqref{eq:ineq} follow from the positive-definiteness of the probability, the former being implied by the latter and by Eq.~\eqref{eq:zeroav}. 
\end{proof}
\section{Conditions of validity for the theorem of Paw{\l}owski et al.}\label{app:pawlth}
Reference \cite{Pawlowski2010} proves that a family of hidden variable theories 
requires one of the party to have information about both the remote setting and the remote outcome. 
In the following, we clarify the assumptions actually made in Ref.~\cite{Pawlowski2010}, 
and show that the models discussed in the present manuscript do not comply with these assumptions, 
so that there is no contradiction. 
First of all, we notice that two hypotheses are made explicitly: ``freedom of choice" and ``realism". 
The second hypothesis is but counterfactual-definiteness, i.e. the existence of a master probability 
$P(A_0,A_1,B_0,B_1|a_0,a_1,b_0,b_1)$ such that the observed probability $P(A_j,B_k|a_j,b_k)$ 
for any two settings $a_j,b_k$ 
is its marginal. It is well known \cite{Pitowsky1994,DiLorenzo2012d} that the hypothesis of counterfactual-definiteness 
alone is sufficient in order to derive Bell-type inequalities. 
Thus, Ref.~\cite{Pawlowski2010} is not actually using the hypothesis of ``realism'', or the models considered 
could not possibly violate the CHSH inequality.  This leaves only the hypothesis of ``freedom of choice", 
which is akin to what in the present paper is referred to as ``Uncorrelated Choice" (or Measurement Independence). 
There is a difference, however, in that the ``freedom of choice" used in Ref.~\cite{Pawlowski2010} refers 
to freedom only within two possible choices.\footnote{As in principle each detector can take any orientation along the unit 
sphere, calling ``freedom of choice" this hypothesis is somewhat misleading.}
Furthermore, in addition to this hypothesis, Ref. \cite{Pawlowski2010} makes other assumptions that are sparse in the text and not stated as hypotheses, and the conclusions are not restricted to the models satisfying said assumptions. 
Let us enunciate all the hypotheses actually made:  
\begin{enumerate}
\item $A$ and $B$ are limited to two choices each $\mathbf{a}_j,\mathbf{b}_k$, $j,k=0,1$. 
\item Within this restriction, the choices are not influenced by the hidden parameters, and \emph{vice versa}, so that 
$p(j,k|\lambda)=p(j,k)=1/4$ and $\mu(\lambda|j,k)=\mu(\lambda)$. 
\item $A$ and $B$ are mimicking the results of a measurement, they are not actually performing one. 
To this goal they are sharing an information $\lambda$. 
\item
$B$ gives an output $\tau$ according to an algorithm that provides a number, $0\le P^B_+(\lambda,k)\le 1$: if a random number 
between $0$ and $1$ is larger than $P^B_+(\lambda,k)$, $B$ will output $\tau=-1$, otherwise $\tau=+1$.
\item $A$ receives an information $X$ from $B$, in addition to $\lambda$, and tries to mimic 
the conditional probability $P_\sigma(\lambda,\mathbf{a},\mathbf{b},\tau)$ by an algorithm 
providing a threshold $P^A_+(\lambda,j,X)$. 
\end{enumerate}
Reference \cite{Pawlowski2010} demonstrates that under hypotheses (1)-(5), if the CHSH inequality is violated then 
$\max_k\{Prob(k|\lambda,X)\}>1/2$ and $\max_\tau\{Prob(\tau|\lambda,X)\}>1/2$, at least for some $\lambda,X$. 
Hypothesis (1) is crucial: in the Toner and Bacon model, the information $X=c$
does not allow to extract any information about the remote outcome, if $\mathbf{a},\mathbf{b}$ can vary 
over the whole unit sphere. 
Nevertheless, the models presented herein are valid for any distribution of the settings, and they can 
be restricted to two binary choices of polarizations. Hence, the reason of the apparent discrepancy 
does not reside in hypothesis (1). 
The key, instead, is hypothesis (5): the observer at $A$ is not measuring a physical property of a system, but is 
calculating a number through an algorithm, which receives $\lambda,X$ as an input, 
and mimicking a conditional probability accordingly. 
By contrast, let us see how $A$ would estimate the conditional probability from the experimental data if 
a measurement was actually performed: First, 
$A$ and $B$ make a large number of measurements. They 
disclose the settings $\mathbf{a},\mathbf{b}$ and the outcomes $\sigma,\tau$ that they used and observed in each 
individual trial. Then $A$ selects the data for which, say, $\mathbf{b}=\mathbf{b}_0$ and $\tau=\tau_0$, and 
estimates the conditional probability of obtaining $\sigma$ with the frequency that was observed in this subset of data. 
Thus, both setting and outcome information must be sent to $A$ in order to extract the conditional probabilities.  
The results of Ref. \cite{Pawlowski2010} put some restrictions on the models that try to reproduce the 
conditional probability through an algorithm, but do not affect models, like the ones introduced in our paper, 
that assume the existence of additional parameters $\lambda$ giving a finer description of a physical system. 
In this case, the outcome and setting information needs to be sent only after the measurements have been performed, 
by the very definition of conditional probability.

\end{document}